\documentclass[11pt,a4paper,final]{article}
\usepackage[utf8]{inputenc}
\usepackage[T1]{fontenc}
\usepackage{amsmath}
\usepackage{amsfonts}
\usepackage{amssymb}
\usepackage{enumerate}
\usepackage{graphicx}
\usepackage[rightcaption]{sidecap}
\usepackage{import}
\usepackage{textcomp}
\usepackage{xcolor}
\usepackage{bbm}
\usepackage{slashed}
\usepackage{subfigure}
\usepackage{lipsum}
\usepackage{bm}
\usepackage{mathrsfs}
\usepackage[ruled,linesnumbered]{algorithm2e}

\newcommand{\bbR}{\mathbb{R}}

\newcommand{\calK}{\mathcal{K}}

\SetKwComment{Comment}{/* }{ */}
\SetKwInOut{Input}{Input}
\SetKwInOut{Output}{Output}

\usepackage{amsthm}

\theoremstyle{definition}
\newtheorem{theorem}{Theorem}
\newtheorem{lemma}[theorem]{Lemma}

\theoremstyle{remark}
\newtheorem{remark}{Remark}

\author{Adrian Wiltz, Dimos V. Dimarogonas}
\title{A Note on an Upper-Bound for the Sum of a Class $ \calK $ and an Extended Class $ \calK_{e} $ Function
\thanks{This work was supported by the ERC Consolidator Grant LEAFHOUND, the Horizon Europe EIC project SymAware (101070802), the Swedish Research Council, and the Knut and Alice Wallenberg Foundation, and is part of a journal paper that is in preparation.}
\thanks{The authors are with the Division of Decision and Control Systems, KTH Royal Institute of Technology, SE-100 44 Stockholm, Sweden
	{\tt\small \{wiltz,dimos\}@kth.se}.}
}
\date{}		

\renewcommand\footnotemark{}

\begin{document}
\setlength{\parskip}{0.5em}
\maketitle

In this note, we derive an upper-bound for the sum of two comparison functions, namely for the sum of an extended class~$ \calK_{e} $ function $ \alpha_{1} $ and a class~$ \calK $ function $ \alpha_{2} $. A class~$ \calK $ function is defined as a continuous, strictly increasing function $ \alpha_{2}: \bbR_{\geq0}\rightarrow\bbR_{\geq0} $ with $ \alpha(0) = 0 $. If a class $ \calK $ function is additionally defined on the entire real space $ \bbR $, then the function is called an extended class $ \calK_{e} $ function. In particular, an extended class $ \calK_{e} $ function is defined as a continuous, strictly increasing function $ \alpha_{1}: \bbR\rightarrow\bbR $ with $ \alpha(0) = 0 $.

Class~$ \calK $ and extended class~$ \calK_{e} $ functions are a particular type of the class of comparison functions. Comparison functions play an important role for example in the analysis of nonlinear dynamic systems and in nonlinear controller design~\cite{Khalil2002}. A notable collection of results on comparison functions can be found in~\cite{Kellett2014}. However, to the best of our knowledge, the relations derived in this note have not been previously derived in the literature.

This note is dedicated to the derivation of a super-additivity like upper-bound on the sum of an extended class~$ \calK_{e} $ function $ \alpha_{1} $ and a class~$ \calK $ function $ \alpha_{2} $. More precisely, we show under which conditions the super-additivity like property
\begin{align*}
	\alpha_{1}(x_{1}) + \alpha_{2}(x_{2}) \leq \beta(x_{1}+x_{2})
\end{align*}
holds where $ \beta: \bbR\rightarrow\bbR $ is an extended class~$ \calK_{e} $ function. 

Our results are stated as follows.

\begin{lemma}
	\label{lemma:time-varying CBF without input constraints 1}
	Let $ \alpha_{1}: \bbR \rightarrow \bbR $ be an extended class $ \calK_{e} $ function, and $ \alpha_{2}: \bbR_{\geq 0} \rightarrow \bbR $ a \emph{convex} class $ \calK $ function such that $ \alpha_{1}(-x)\leq -\alpha_{2}(x) $ for all $ x\in[0,A] $ and some finite $ A>0 $. Then, there exists an extended class~$ \calK_{e} $ function $ \beta $ such that for all $ x_{1} \in [-A,\infty) $, $ x_{2} \in [0,A] $ it holds
	\begin{align}
		\label{eq:time-varying CBF without input constraints 1}
		\alpha_{1}(x_{1}) + \alpha_{2}(x_{2}) \leq \beta(x_{1}+x_{2}).
	\end{align}
\end{lemma}

\begin{lemma}
	\label{lemma:time-varying CBF without input constraints 2}
	Let $ \alpha_{1}: \bbR \rightarrow \bbR $ be an extended class $ \calK_{e} $ function, and $ \alpha_{2}: \bbR_{\geq 0} \rightarrow \bbR $ a \emph{concave} class $ \calK $ function such that $ \alpha_{1}(-x)\leq -\alpha_{2}(x) $ for all $ x\in[0,A] $ and $ A>0 $. Then, there exists an extended class~$ \calK_{e} $ function $ \beta $ such that for all $ x_{1} \in [-A,\infty) $, $ x_{2} \in [0,A] $ it holds
	\begin{align}
		\label{eq:time-varying CBF without input constraints 2}
		\alpha_{1}(x_{1}) + \alpha_{2}(x_{2}) \leq \beta(x_{1}+x_{2}).
	\end{align}
	This even holds if $ A\rightarrow\infty $.
\end{lemma}

\begin{remark}
	In contrast to Lemmas~\ref{lemma:time-varying CBF without input constraints 1} and~\ref{lemma:time-varying CBF without input constraints 2}, the compendium of comparison function results~\cite{Kellett2014} only reports results where $ \beta(x_{1}+x_{2}) $ constitutes a lower-bound. 
\end{remark}

In the remainder of this note, we prove these two lemmas.

\newpage

\begin{proof}[Proof of Lemma~\ref{lemma:time-varying CBF without input constraints 1}]
	Before we start with the actual proof, we recall some important properties of convex functions which the proof is based on. At first, recall that for a convex function $ \alpha' $ it holds for all $ x,y\in\bbR $ that 
	\begin{align}
		\label{eq:time-varying CBF without input constraints 1 aux 00}
		\alpha'(\sigma x \!+\! (1\!-\!\sigma)y) \!\leq\! \sigma \alpha'(x) \!+\! (1\!-\!\sigma) \alpha'(y), \quad \!\sigma\!\in\![0,1].
	\end{align}
	For $ \sigma=1 $, this implies $ \alpha'(\sigma x) \leq \sigma \alpha'(x) $. Moreover, if additionally $ \alpha'(0)=0 $ (e.g., if $ \alpha' $ is a convex class~$ \calK $ function), then 
	$ \alpha' $ is superadditive for positive real numbers; that is, for all $ x,y\geq 0 $, it holds that
	\begin{align}
		\label{eq:time-varying CBF without input constraints 1 aux 0}
		\alpha'(x)+\alpha'(y) \leq \alpha'(x+y).
	\end{align}
	This can be shown as
	\begin{align*}
		\alpha'(x) &+ \alpha'(y) 
		= \alpha'\left( (x+y) \frac{x}{x+y} \right) + \alpha'\left( (x+y) \frac{y}{x+y} \right) \\
		&\leq \frac{x}{x+y} \alpha'(x+y) + \frac{y}{x+y} \alpha'(x+y) = \alpha'(x+y)
	\end{align*}
	where the inequality results from the fact that $ \alpha'(\sigma x) \leq \sigma \alpha'(x) $ for $ \sigma\in[0,1] $.
	
	Furthermore, we recall that the difference quotient $ D_{\alpha'}(x,y) := \frac{\alpha'(y)-\alpha'(x)}{y-x} $ of the convex function $ \alpha' $ is monotonously increasing in both of its arguments\footnote{This is a standard result and it can be easily shown as follows. At first, let $ x $ be fixed, and choose $ y = \sigma y' + (1-\sigma)x  $ where $ y'>x $; thus, $ y \leq y' $. Then, it follows that $ D(x,y) $ is monotonously increasing in~$ y $ as 
		\begin{align*}
			&D(x,y) = \frac{\alpha'(y)-\alpha'(x)}{y-x} = \frac{\alpha'(\sigma y' + (1-\sigma)x) - \alpha'(x)}{\sigma y' + (1-\sigma)x - x}  \\
			&\quad= \frac{\alpha'(\sigma y' + (1-\sigma)x) - \alpha'(x)}{\sigma y' - \sigma x}\stackrel{\eqref{eq:time-varying CBF without input constraints 1 aux 00}}{\leq} \frac{(1-\sigma)\alpha'(x) + \sigma \alpha'(y') - \alpha'(x)}{\sigma y' - \sigma x} \\
			&\quad= \frac{\alpha'(y')-\alpha'(x)}{y'-x} = D(x,y').
		\end{align*}
		The result follows for the first argument analogously. 
	}. Thus, we have $ \frac{\alpha'(y)-\alpha'(x)}{y-x} \leq \frac{\alpha'(y+c)-\alpha'(x+c)}{y-x} $ for all $ x<y $, $ c\geq 0 $, or equivalently,
	\begin{align}
		\label{eq:time-varying CBF without input constraints 1 aux 1}
		\alpha'(y)-\alpha'(x) \leq \alpha'(y+c)-\alpha'(x+c).
	\end{align}
	
	At last in addition to its convexity, we assume that $ \alpha' $ is an extended class~$ \calK_{e} $ function. Then, it holds for all $ \sigma \in [0,1/2] $, $ x\geq 0 $, that $ 0 \leq \alpha'((1-2\sigma)x) \leq \sigma \alpha'(-x)  + (1-\sigma) \alpha'(x) $ where the first inequality follows as $ \alpha' $ is class~$ \calK_{e} $ and thus it is non-negative for non-negative arguments; the second inequality follows from the convexity of $ \alpha' $. By rearranging terms, we obtain $ -\alpha'(x)\leq \sigma (\alpha'(-x)-\alpha'(x)) $, and thus for $ \sigma = 1/2 $ that for all $ x\geq 0 $ it holds
	\begin{align}
		\label{eq:time-varying CBF without input constraints 1 aux 2}
		-\alpha'(x)\leq\alpha'(-x).
	\end{align} 
	
	Now, we turn towards the actual proof of \eqref{eq:time-varying CBF without input constraints 1}. Therefore, let us define an extended version of $ \alpha_{2} $ as a convex extended class $ \calK_{e} $ function $ \alpha'_{2}: \bbR\rightarrow\bbR $ such that: (1) $ \alpha'_{2}(x) $ is an arbitrary continuous, convex and monotonously increasing continuation of $ \alpha_{2}(x) $ for all $ x<0 $; (2) $ \alpha'_{2}(x) = \alpha_{2}(x) $ for all $ x\in[0,A] $; and (3) $ \alpha'_{2}(x) = \alpha_{2}(A) + \alpha'_{1}(x-A) $ for all $ x\geq A $ where $ \alpha'_{1}(x) $ is some convex class~$ \calK $ function with $ \alpha'_{1}(x)\geq\alpha_{1}(x) $ for all $ x\geq0 $. 
	Next, we distinguish three cases, namely $ x_{1} \in [-A,0] $ with $ x_{1}+x_{2}\leq 0 $ (case~1a) and $ x_{1}+x_{2}\geq 0 $ (case~1b), and $ x_{1} \in [0,\infty) $ (case~2). Recall that $ x_{2}\in[0,A] $.
	
	\emph{Case~1a ($ x_{1} \in [-A,0] $ and $ x_{1}+x_{2}\leq0 $):} At first we note that since $ \alpha_{2} $ is convex it holds 
	\begin{align}
		\label{eq:time-varying CBF without input constraints 1 aux 3}
		\alpha_{2}(-(x_{1}+x_{2})) + \alpha_{2}(x_{2}) \stackrel{\eqref{eq:time-varying CBF without input constraints 1 aux 0}}{\leq} \alpha_{2}(-x_{1})
	\end{align}
	due to the superadditivity of $ \alpha_{2} $. Next, we consider the left-hand side of~\eqref{eq:time-varying CBF without input constraints 1}. By employing that $ \alpha_{1}(-x)\leq -\alpha_{2}(x) $ for all $ x\in[0,A] $ and that $ \alpha'_{2} $ is convex, we obtain
	\begin{align*}
		\alpha_{1}(x_{1}) \!+\! \alpha_{2}(x_{2}) \!&\leq\! -\alpha_{2}(-x_{1}) \!+\! \alpha_{2}(x_{2}) \!\stackrel{\eqref{eq:time-varying CBF without input constraints 1 aux 3}}{\leq}\! -\alpha_{2}(-(x_{1}\!+\!x_{2})) \\
		&= -\alpha'_{2}(-(x_{1}+x_{2})) \stackrel{\eqref{eq:time-varying CBF without input constraints 1 aux 2}}{\leq} \alpha'_{2}(x_{1}+x_{2}).
	\end{align*}
	Here we employed that $ -(x_{1}+x_{2})\in[0,A] $ in case~1a.
	
	\emph{Case~1b ($ x_{1} \in [-A,0] $ and $ x_{1}+x_{2}>0 $):} Noting that $ x_{1} + A \geq 0 $, we derive by starting again with the left-hand side of~\eqref{eq:time-varying CBF without input constraints 1} that
	\begin{align*}
		\alpha_{1}(x_{1}) \!\!+\!\! \alpha_{2}(x_{2}) \!&\leq\! -\alpha_{2}(-x_{1}) \!\!+\!\! \alpha_{2}(x_{2}) \!=\!  -\alpha'_{2}(-x_{1}) \!\!+\!\! \alpha'_{2}(x_{2}) \\ 
		\!&\stackrel{\eqref{eq:time-varying CBF without input constraints 1 aux 1}}{\leq}\! -\alpha'_{2}(-x_{1}\!+\!x_{1}\!+\!A) \!+\! \alpha'_{2}(x_{1}\!+\!x_{2}\!+\!A) \\
		&= -\alpha'_{2}(A) \!+\! \alpha'_{2}(x_{1}\!+\!x_{2}\!+\!A).
	\end{align*}
	
	\emph{Case~2 ($ x_{1} \in [0,\infty) $):} Recall that $ \alpha'_{2}(x) = \alpha_{2}(A) + \alpha'_{1}(x-A) $ for $ x\geq A $ with $ \alpha'_{1}(x)\geq\alpha_{1}(x) $ for all $ x\geq0 $. Thus, we have for $ x\geq 0 $ that
	\begin{align}
		\label{eq:time-varying CBF without input constraints 1 aux 4}
		\alpha_{1}(x) \leq \alpha'_{1}(x) = \alpha'_{2}(x+A) - \alpha_{2}(A).
	\end{align} 
	Furthermore by employing that $ \alpha'_{2} $ is convex, we obtain
	\begin{align*}
		\alpha_{1}(x_{1}) \!+\! \alpha_{2}(x_{2}) \!&\stackrel{\eqref{eq:time-varying CBF without input constraints 1 aux 4}}{\leq}\! - \alpha_{2}(A) \!+\! \alpha'_{2}(x\!+\!A) \!+\! \alpha_{2}(x_{2}) \\ 
		&= \alpha'_{2}(x\!+\!A) \!-\! \alpha'_{2}(A) \!+\! \alpha'_{2}(x_{2}) \\ &\stackrel{\text{\eqref{eq:time-varying CBF without input constraints 1 aux 0}}}{\leq} -\alpha'_{2}(A) \!+\! \alpha'_{2}(A\!+\!x_{1}\!+\!x_{2}).
	\end{align*}
	
	Summarizing cases~1a, 1b and~2, we choose $ \beta $ in~\eqref{eq:time-varying CBF without input constraints 1} as 
	\begin{align*}
		\beta(x_{1}\!+\!x_{2}) = \begin{cases}
			\alpha'_{2}(x_{1}\!+\!x_{2}) &\text{if } x_{1}\!+\!x_{2}\leq0,  \\
			-\alpha'_{2}(A)\!+\!\alpha'_{2}(x_{1}\!+\!x_{2}\!+\!A) &\text{if } x_{1} \!+\! x_{2}> 0.
		\end{cases}
	\end{align*}
	As shown, $ \beta $ satisfies~\eqref{eq:time-varying CBF without input constraints 1}. Moreover, $ \beta $ is continuous, it holds $ \beta(0) = 0 $, and $ \beta $ is monotonously increasing as $ \alpha'_{2} $ is class $ \calK_{e} $; thus, also $ \beta $ is class~$ \calK_{e} $. This concludes the proof.
\end{proof}

\begin{proof}[Proof of Lemma~\ref{lemma:time-varying CBF without input constraints 2}]
	The proof in the case of a concave function $ \alpha_{2} $ is more straightforward compared to the convex case. Before we start, we recall that the difference quotient $ D_{\alpha'}(x,y) := \frac{\alpha'(y)-\alpha'(x)}{y-x} $ of a concave function $ \alpha' $ is monotonously decreasing in both arguments. Thus, \eqref{eq:time-varying CBF without input constraints 1 aux 1} still holds, however only for non-positive $ c\leq 0 $. More precisely, it holds for all $ x<y $, $ c\leq 0 $ that
	\begin{align}
		\label{eq:time-varying CBF without input constraints 2 aux 1}
		\alpha'(y)-\alpha'(x) \leq \alpha'(y+c)-\alpha'(x+c).
	\end{align}
	
	Now, we turn towards the actual proof of~\eqref{eq:time-varying CBF without input constraints 2}. To this end, we define again an extended version of $ \alpha_{2} $, however this time as a concave extended class $ \calK_{e} $ function $ \alpha'_{2}:\bbR\rightarrow\bbR $ such that $ \alpha'_{2}(x)=\alpha_{2}(x) $ for all $ x\geq 0 $. Next, we distinguish two cases, namely $ x_{1}\in[-A,0] $ (case~1) and $ x_{1}\in[0,\infty) $ (case~2). Recall that $ x_{2}\in[0,A] $.
	
	\emph{Case 1 ($ x_{1}\in[-A,0] $):} Consider the left-hand side of~\eqref{eq:time-varying CBF without input constraints 2}. By employing that $ \alpha_{1}(-x)\leq -\alpha_{2}(x) $ for all $ x\in[0,A] $ and that $ \alpha'_{2} $ is concave, we obtain
	\begin{align*}
		\alpha_{1}(x_{1}) \!&+\! \alpha_{2}(x_{2}) \!\leq\! -\alpha_{2}(-x_{1}) \!\!+\!\! \alpha_{2}(x_{2}) \!=\! -\alpha'_{2}(-x_{1}) \!\!+\!\! \alpha'_{2}(x_{2}) \\
		&\stackrel{\eqref{eq:time-varying CBF without input constraints 2 aux 1}}{\leq} -\alpha'_{2}(-x_{1} \!+\! x_{1}) \!+\! \alpha'_{2}(x_{1}\!+\!x_{2}) = \alpha'_{2}(x_{1}\!+\!x_{2})
	\end{align*}
	where the last inequality is obtained by adding $ c=x_{1} $, which is non-positive by assumption, to the arguments of $ \alpha'_{2} $. For the case that $ x_{1} + x_{2}\geq 0 $, we note that the right-hand side is upper-bounded by
	\begin{align*}
		\alpha'_{2}(x_{1}+x_{2}) \leq \alpha_{1}(x_{1}+x_{2})+\alpha'_{2}(x_{1}+x_{2}).
	\end{align*}
	This observation is needed later on for the construction of the extended class~$ \calK_{e} $ function~$ \beta $.
	
	\emph{Case 2 ($ x_{1}\in[0,\infty) $):}  In this case, it always holds that $ x_{1}+x_{2}\geq 0 $. Thus, we directly obtain 
	\begin{align*}
		\alpha_{1}(x_{1}) \!+\! \alpha_{2}(x_{2}) &= \alpha_{1}(x_{1}) \!+\! \alpha'_{2}(x_{2}) \\
		&\leq \alpha_{1}(x_{1}\!+\!x_{2}) \!+\! \alpha'_{2}(x_{1}\!+\!x_{2}).
	\end{align*}
	
	Summarizing cases~1 and~2, we choose $ \beta $ in~\eqref{eq:time-varying CBF without input constraints 2} as
	\begin{align*}
		\beta(x_{1}\!+\!x_{2}) = 
		\begin{cases}
			\alpha'_{2}(x_{1}\!+\!x_{2}) &\text{if } x_{1}\!+\!x_{2}< 0, \\
			\alpha_{1}(x_{1}\!+\!x_{2}) \!+\! \alpha'_{2}(x_{1}\!+\!x_{2}) &\text{if } x_{1}\!+\!x_{2}\geq 0.
		\end{cases}
	\end{align*}
	As shown, $ \beta $ satisfies~\eqref{eq:time-varying CBF without input constraints 2}. Moreover, $ \beta $ is continuous, it holds $ \beta(0) = 0 $, and $ \beta $ is monotonously increasing as both $ \alpha_{1} $ and $ \alpha_{2}' $ are class $ \calK_{e} $ functions; thus, also $ \beta $ is a class~$ \calK_{e} $ function. This result even holds for $ A\rightarrow\infty $, as the construction of $ \beta $ does not rely on $ A $. This concludes the proof.
\end{proof}


\bibliographystyle{IEEEtrans}
\bibliography{/Users/wiltz/CloudStation/JabBib/Research/000_MyLibrary}

\end{document}